\documentclass[pra,twocolumn,superscriptaddress]{revtex4-1}

\usepackage{graphicx}
\usepackage{mathptmx}
\usepackage{amsmath}
\usepackage{amsfonts}
\usepackage{dsfont}
\usepackage{amssymb}
\usepackage{amsthm}
\usepackage{color}
\usepackage{tikz}
\usetikzlibrary{positioning}
\usetikzlibrary{patterns}
\usetikzlibrary{calc}
\usepackage{dsfont}
\usepackage{hyperref}
\usepackage[off]{auto-pst-pdf}

\newcommand{\tensormax}{\otimes_\mathrm{max}}
\newcommand{\gtensormax}{\overline{\otimes}_\mathrm{max}}
\newcommand{\tensormin}{\otimes_\mathrm{min}}
\newcommand{\Tr}{\mathrm{Tr}}
\newcommand{\id}{\mathds{1}}

\newtheorem{thm}{Theorem}

\newtheorem{prop}[thm]{Proposition}
\theoremstyle{remark}

\theoremstyle{definition}
\newtheorem{dfn}[thm]{Definition}

\definecolor{structure}{rgb}{0.2,0.3,0.7}

\begin{document}

\title{Generalized Probabilistic Theories Without the No-Restriction Hypothesis}

\author{Peter Janotta}
\affiliation{Universit\"{a}t W\"{u}rzburg, Am Hubland, Fakult{\"{a}}t f\"{u}r Physik und Astronomie, 97074 W\"{u}rzburg, Germany}
\author{Raymond Lal}
\affiliation{University of Oxford, Department of Computer Science, Quantum Group, Wolfson Building, Parks Road, Oxford OX1 3QD, UK.}

\pacs{03.65.-w}
% Explanation of PACS numbers:
% 03.65.-w      Quantum mechanics
\keywords{Foundations of quantum theory \and generalized probabilistic theories \and tensor products}

\begin{abstract}
The framework of generalized probabilistic theories (GPTs) is a popular approach for studying the physical foundations of quantum theory.
 The standard framework assumes the no-restriction hypothesis, in which the state space of a physical theory determines the set of measurements. 
 However, this assumption is not physically motivated. 
 We generalize the framework to account for systems that do not obey the no-restriction hypothesis. 
 We then show how our framework can be used to describe new classes of probabilistic theories, for example those which include intrinsic noise. 
Relaxing the restriction hypothesis also allows us to introduce a `self-dualization' procedure, which yields a new class of theories that share many features of quantum theory, such as obeying Tsirelson's bound for the maximally entangled state.
 We then characterize joint states, generalizing the maximal tensor product.
 We show how this new tensor product can be used to describe the convex closure of the Spekkens toy theory, and in doing so we obtain an analysis of why it is local in terms of the geometry of its state space.
We show that the unrestricted version of the Spekkens toy theory is the theory known as `boxworld' that allows maximal nonlocal correlations.
\end{abstract}

\maketitle

%\tableofcontents

\section{Introduction}\label{sec:intro}

The framework of generalized probabilistic theories (GPTs) is a modern operational approach for studying the physical foundations of quantum theory \cite{Barrett07}.
The framework is \em operational \em because a theory is defined according to the observable measurement statistics that it predicts.
In contrast, quantum theory is usually defined using an abstract mathematical formalism without physical motivation (e.g. the density matrix formalism).
Assuming only basic principles, the framework encompasses a large variety of theories. 
For example, quantum theory and classical probability theory are special cases of GPTs. 
The focus of work on GPTs is to identify the unique physical properties that distinguish quantum theory from other theories.
More generally, one can examine the relationship between different physical properties, such as no-cloning and nonlocality, without restricting to a particular physical theory.

Using this framework, it has been shown that many properties that were thought to be particular to quantum theory are in fact very general.
As a sample of such results, it was shown that any non-classical probability theory (in the sense to be described in section \ref{sec:GPTs}) has the following properties:  the existence of entanglement \cite{Barrett07}; for mixed states, the lack a unique decomposition into a unique ensemble of pure states; generalizations of the no-cloning or no-broadcasting theorem \cite{Barnum08}; and, an information-disturbance trade-off \cite{Scarani06}. 
Notably, recent attempts to reconstruct quantum theory from physical axioms include the assumptions made in GPTs \cite{Masanes11} or very similar assumptions \cite{Hardy11,Chiribella11}.

A GPT is defined by a set of preparations, a set of measurements, and composition rules for multipartite systems called the tensor product of the theory.
In general there is a trade-off between possible preparations and possible measurement outcomes: the larger the set of preparations, the smaller the upper bound on the set of allowed measurements \cite{Short2010}.
In the existing GPT framework, it is usually assumed that this upper bound is \em saturated\em.
This means that, for a chosen set of states, \em all \em potential measurement outcomes that yield probability-valued results are assumed to be physically realizable.
This is called the \em no-restriction hypothesis \em \cite{Chiribella11}.
This assumption is not based on any physical motivation, and it is usually assumed for the sake of mathematical convenience. 
In this work we take on the task of extending the framework of GPTs when the no-restriction hypothesis is abandoned.
This extension of GPTs therefore brings the framework closer to the operational motivation for which it was originally initiated.

\paragraph*{Our contribution.}
The idea of removing the no-restriction hypothesis (or replacing it with other assumptions) has appeared sporadically in other works \cite{Chiribella11,barnum10}.
However, until now  a systematic analysis of the consequences of doing so has been lacking. 
In this paper we provide a well-defined framework with the no-restriction hypothesis omitted, whilst keeping the other assumptions of the GPT framework.
Our work then proceeds in two parts.

In the first part we show that this new framework encompasses more theories than before. 
For example, we show that theories with intrinsic noise can be described in our framework, but not in the existing GPT framework.
We also provide a procedure for constructing a \em self-dual \em theory from a theory which is not self-dual.
The importance of this is that self-duality has been shown to imply `quantum-like' (for example, limiting bipartite nonlocality to Tsirelson's bound for the maximally entangled state \cite{polypaper}).
Hence this allows us to introduce a new class of probabilistic theories with `quantum-like' behaviour, and crucially, this is a class of theories which does not satisfy the no-restriction hypothesis.

In the second part, we develop the treatment of composite systems.
In particular, we show that our extension requires a  new (and more general) definition of the tensor product for describing composite systems. 
This significantly extends the GPT framework, since it allows us to analyse the relationship between nonlocality and the geometry of the state space of a theory, building on previous work in this direction.
For example, we show how  the Spekkens toy theory (for which the connection to GPTs had not been previously established) can be viewed as a GPT, but only in our more general framework.
Moreover, this allows us to give an analysis of why the Spekkens theory is local, using the geometry of its state space.

\paragraph*{Structure of the paper.}
In section \ref{sec:GPTs} we give a brief overview of the framework of GPTs.
We then begin the first part of our analysis, concentrating on single systems.
In section \ref{sec:norestriction} we describe in detail the no-restriction hypothesis, and some consequences of relaxing this assumption. 
In section \ref{sec:noisyboxworld} we develop the important example of theories with noise.
In section \ref{sec:polymodel} we introduce the self-dualization procedure, and discuss the class of theories that this introduces.
We then enter the second part of our analysis, which concerns composite systems.
In section \ref{sec:jointstates} we explain how joint states of composite systems are usually described.
In section \ref{sec:genmaxtensor} we show why a new definition of composite systems is needed, and we introduce this definition. 
We then study examples of theories such as the Spekkens model.

\section{Generalized probabilistic theories: a brief summary}\label{sec:GPTs}

A physical experiment consists of the following steps: the preparation of a system, transformations of that system (e.g. by inherent dynamics), and a measurement.
In general, the measurement will different outcomes, each occurring with some probability.
Defining a generalized probabilistic theory amounts to specifying these probabilities for any such combination of preparation, transformation and measurement.
Note that transformations can be absorbed into either the preparation or the measurement. 
Hence to define the allowed probability distributions of a GPT, it suffices to define the set of preparation procedures and the set of measurements.

\subsubsection{States and effects}

Consider a class of preparation procedures which all yield exactly the same measurement statistics.
The members of this class are experimentally indistinguishable.
Since a GPT concerns only experimental statistics, we can define a \em  state \em of a system as such an equivalence class.
Analogously we also define an \em effect \em as an equivalence class of measurement outcomes. 
We will refer to this identification of states and effects with their respective measurement statistics as the \em equivalence principle\em.
Mathematically, states are represented by elements of a vector space $V$. 
Effects are linear functionals on states, i.e. elements of the dual space $V^*$.
Applying an effect $e$ to a state $\omega$ yields the probability $p(e|\omega) = e(\omega)$ for the corresponding measurement outcome to occur when measuring the system in the state.
Without loss of generality we will choose a specific representation of states and effects in this paper to demonstrate the abstract concepts.
Both states and effects will be represented by vectors embedded in $\mathbb{R}^n$.
The application of effects on states is given by the Euclidean inner product of the respective vectors:
\begin{align}
 \label{eq:vecrep}
 e = \left( \epsilon_1, \cdots , \epsilon_n\right)^T \qquad \omega = \left( w_1 , \cdots , w_n\right)^T\\
 p(e|\omega) = e^T \!\!\!\cdot \omega = \sum_i \epsilon_i \, w_i
\end{align}

The GPT framework also accounts for ensembles of preparations or measurements, in which there is uncertainty about which measurement is implemented, or which state has been prepared.
This could occur if there is a probabilistic selection of the preparation procedure, for example.
This probability distribution is represented by using \em mixed states \em and \em mixed effects\em, given by convex combinations:
\begin{align}
 e &= \sum_i \lambda_i \, e_i \qquad \lambda_i \geq 0, \, \sum_i \lambda_i = 1\\
 \omega &= \sum_i \mu_j \, \omega_j \qquad \mu_i \geq 0, \, \sum_i \mu_i = 1
\end{align}
corresponding to ensembles $\{\lambda_i, e_i\}$ and $\{\mu_i, \omega_i\}$.
Consequently, states and effects form convex sets. 
If the only convex decomposition of a state $\omega$ is such that $\omega\propto \omega_i$ for all $i$, then the state is a \em pure state\em.
Similarly, if the only convex decomposition of an effect $e$ using Eq.~\ref{eq:vecrep}  is such that $e \propto e_i$ for all $i$, then the effect is a \em pure effect\em.

Since effects and states act linearly on each other, the probability distribution for the ensembles is the weighted sum of the probabilities $p_{ij}=e_i(\omega_j)$ of individual ensemble elements:
\begin{align}
 e(\omega) = \sum_{i,j} \lambda_i \, \mu_j \, e_i(\omega_j).
\end{align}
More generally, consider the result of applying different measurements to systems prepared by the same method. 
In general, there will be measurement outcomes with probabilities that are linearly dependent for a fixed state.
Analogously, one might find linear dependencies between the probabilities for a fixed measurement outcome under variations of the state that is prepared.
This implies a linear dependence between the vector space elements $\omega\in V$ representing the states; there is a corresponding linear dependence for the effects $e\in V^*$. 
This determines the \em dimension \em of $V$ as the minimal number of different measurement outcomes needed to identify a state uniquely (this is called the `fiducial set' of measurement outcomes by Hardy \cite{Hardy}). 
In this paper we restrict ourselves to systems for which the vector space $V$ has finite dimension.
Hence the dimension of $V$ is equal to the dimension of the dual space $V^*$, which is the minimal number of preparations required to identify an effect.

\subsubsection{Normalization and measurements}
A central concept in the GPT framework is the description of \textit{perfect} preparations and measurements.
A perfect preparation is one that is guaranteed to succeed.
It is represented by a normalized state, where normalization defined with respect to a special effect, called the \em unit measure \em $u$.
The set of all normalized states is called the \em state space \em $\Omega$. 
The unit measure $u$ represents an unbiased measurement with only one outcome: this outcome occurs if a preparation has succeeded, i.e. it is determined by
\begin{align}
 u(\omega) = 1 \quad \forall \omega \in \Omega.
\end{align}
In the specific representation used in this paper we choose 
\[
u := \left( 0, \cdots ,0,1 \right)^T.
\]
Consequently, for a state $\omega$ embedded in an $n$-dimensional vector space $V$, the normalization of $\omega$ is directly apparent from the last component $\omega_n$, i.e. normalized states have $\omega_n = 1$. 

An effect is a map $e:\Omega\rightarrow [0,1]$ that gives a probability when applied to a normalized state $\omega$. 
A perfect measurement consists of a set of effects $\{e_i\}$ which sum up to the unit measure, i.e.:
\[
\sum_i e_i = u.
\]
Thus, measurement probabilities sum up to one for any perfectly-prepared system. 

Beyond the description of perfect preparations and measurements, the GPT framework also accounts for the opposite extreme, namely preparations that always fail or measurement outcomes that never occur no matter which state they are applied to.
The corresponding states and effects are given by the zero elements $\emptyset$ of $V$ and $V^*$ with
\begin{align}
 \emptyset(\omega) = 0 \quad &\forall \omega \in V\\
 e(\emptyset) = 0 \quad &\forall e \in V^*.
\end{align}
Imperfections in preparations yield unnormalized states resulting from the mixture of a normalized state $\omega$ and $\emptyset$. 
Detector deficiencies and bias can be addressed by mixing every effect of a perfect measurement with $\emptyset$ or another common effect.
However, we will show in section \ref{sec:linmaps} that consistency conditions on joint states forbid imperfect measurements.
Consequently, the measurement has to be completed by an additional effect, such that the effects sum up to the unit measure, even though the occurrence of this additional measurement outcome cannot be registered by an experimenter due to detector deficiencies.

\subsubsection{Equivalent Representations}
\label{sec:equivreps}
Consider applying arbitrary bijective linear maps $L^T$ on all effects and the corresponding inverse map $L^{-1}$ on all states. 
This leaves the results from any combination of effects and states invariant, since:
\begin{align}
\label{eq:equivalentsets}
\left(L^T \!\!\!\cdot e\right)\!\!\left[L^{-1} \!\!\!\cdot\omega\right] &= \left(L^T \!\!\!\cdot e\right)^T \!\!\!\cdot L^{-1} \!\!\!\cdot\omega = e^T \!\!\!\cdot L \cdot L^{-1} \!\!\!\cdot\omega = e^T \!\!\!\cdot \omega.
\end{align}
Now, a particular probabilistic theory is associated with a particular state space $\Omega$ and set of effects $E$.
But theories are distinguished only by the different measurement statistics that are possible (as is guaranteed by using the equivalence principle).
Hence if $\Omega$ and $E$ are transformed according to \eqref{eq:equivalentsets}, then the resulting $\Omega'$ and $E'$  define the same theory, since this transformed state space and effect set yield the same measurement statistics.

\begin{figure*}[t]
 \centering
  \begin{tikzpicture}
   \node (statespace){
    \begin{tikzpicture}
     \path (3.6-1.2, -1.2) coordinate (origin);
     \begin{scope}
      \clip (origin) -- +(60:2.4) arc (60:120:2.4) -- cycle;
      \shade[inner color=structure!50!white, outer color=white] (origin) circle (2.4);
     \end{scope}     
     \node at (3.6-1.6,0.2) {$V_+$};
     \draw[->, thick, gray] (origin) -> (3.6+.3,-1.2);
     \draw[->, thick, gray] (origin) -> (3.6-1.2,1.2);
     \draw[line width=1.5pt, structure!50!white] (3.6-1.2-1.03, 0.6)--(3.6-1.2+1.03, 0.6);  
     \draw[->, line width=1.3pt, red!80] (origin) -> (3.6-1.2,0.6) node[right, midway, xshift=-1mm]{$u$};
     \draw[black!60] (3.6-1.2,0.6) +(270:0.4cm) arc (270:360:0.4cm) -- (3.6-1.2,0.6);
     \draw[fill, black!60] (2.4, 0.6) +(315:0.23cm) circle (1.2pt);
     \draw (origin)--+(60:2.7) (origin) -- +(120:2.7);
     \node at (0.9,-0.3) {
      \begin{tikzpicture}
       \draw[fill=structure!50!white] (-0.5,-0.5) [rounded corners=10pt] -- (-0.3,0.5) -- (0.7,0.5)
           [sharp corners] -- (0.7,-0.2) -- (0.6,-0.5) -- (0.3,-0.7)
           [rounded corners=5pt] -- cycle;
       \node at (0.15,-0.1) {$\Omega$};
      \end{tikzpicture}
     };
     \draw[thick, ->] (0.9,0.1) to[out=90] (3.6-1.2-1.03+0.5, 0.6);
    \end{tikzpicture}
   } ;

   \node (LogicalArrow) [scale=3*.9, right of=statespace, node distance=1cm] {$\Leftrightarrow$};

   \node (effectspace) [scale=.9,right of=LogicalArrow, node distance=3cm] {
    \begin{tikzpicture}
     \path (-1.2, -1.2) coordinate (origin);
    \draw[line width=2pt, red, ->](origin)--(-1.2, 0.6);
     \begin{scope}
      \clip (origin)--(1.8-1.2, 0.577*1.8-1.2)--(0.6,1.2)--(-3,1.2)--(-3, 0.577*1.8-1.2)--cycle;
      \shade[inner color=red!50, outer color=white] (origin) circle (2.4); 
     \end{scope}
     \shade[bottom color=gray!50!red, top color=black!5] (origin)--+(60:2.78)--+(120:2.78)--cycle;
     \draw[gray, dashed] (origin) -- +(60:2.7) (origin) -- +(120:2.7);
     \draw[->, thick, gray] (origin) -> (0.4,-1.2);
     \draw[->, thick, gray] (origin) -> (-1.2,1.2);
     \draw[line width=1.4pt, gray!50] (-1.2-1.03, 0.6)--(-1.2+1.03, 0.6);  
     \draw[red, opacity=0.5, fill=red!50] (origin)--++(1.57, 0.9)--++(-1.57, 0.9)--++(-1.57,-0.9)--cycle;
     \draw[black!70] (origin) +(30:0.4cm) arc (30:120:0.4cm) -- (origin);
     \draw[fill, black!70] (origin) +(75:0.2cm) circle (1.25pt);
     \draw (origin) -- +(30:2) (origin) -- +(150:2);
     \draw [dotted](origin) ++(30:2)--++(30:.2) (origin) ++(150:2)--++(150:.2);
     \node at (-1.2, -0.3) {$E$};
     \node at (-2.6,0.4) {$V_+^*$};
    \end{tikzpicture}
   };
   \node (alternative) [scale=2*.9, right of=effectspace, node distance=1.5cm] {vs};

   \node (effectspace2) [scale=.9, right of=alternative, node distance=3cm] {
    \begin{tikzpicture}
     \path (-1.2, -1.2) coordinate (origin);
    \draw[line width=2pt, red, ->](origin)--(-1.2, 0.6);
     \begin{scope}
      \clip (origin)--(1.8-1.2, 0.577*1.8-1.2)--(0.6,1.2)--(-3,1.2)--(-3, 0.577*1.8-1.2)--cycle;
      \shade[inner color=red!50, outer color=white] (origin) circle (2.4); 
     \end{scope}
     \shade[bottom color=gray!50!red, top color=black!5] (origin)--+(60:2.78)--+(120:2.78)--cycle;
     \draw[gray, dashed] (origin) -- +(60:2.7) (origin) -- +(120:2.7);
     \draw[->, thick, gray] (origin) -> (0.4,-1.2);
     \draw[->, thick, gray] (origin) -> (-1.2,1.2);
     \draw[line width=1.4pt, gray!50] (-1.2-1.03, 0.6)--(-1.2+1.03, 0.6);  
     \draw[white, opacity=0.25, fill=white!50] (origin)--++(1.57, 0.9)--++(-1.57, 0.9)--++(-1.57,-0.9)--cycle;
     \draw[red, opacity=0.5, fill=red!50] (origin)--++(.8, 1.1)--++(-.8, 0.7)--++(-.8,-1.1)--cycle;
     \draw[black!70] (origin) +(30:0.4cm) arc (30:120:0.4cm) -- (origin);
     \draw[fill, black!70] (origin) +(75:0.2cm) circle (1.25pt);
     \draw (origin) -- +(30:2) (origin) -- +(150:2);
     \draw [dotted](origin) ++(30:2)--++(30:.2) (origin) ++(150:2)--++(150:.2);
     \node at (-1.2, -0.3) {$E$};
     \node at (-2.6,0.4) {$V_+^*$};
    \end{tikzpicture}
   };
  \end{tikzpicture}
\caption{The construction of the effect set $E$ in the traditional GPT framework with no-restriction hypothesis is shown in the middle. Without the no-restriction hypothesis the definition of the effect set gets a independent part of the theory specification (right picture).}
\end{figure*}
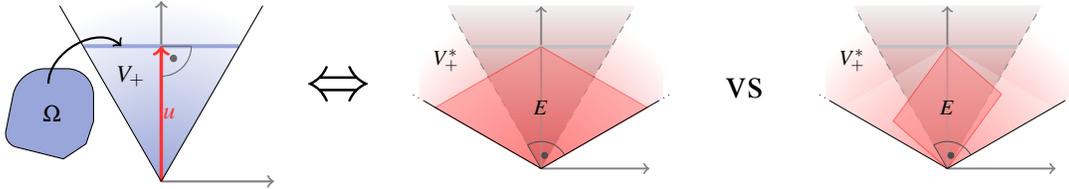

\subsubsection{Examples}
\paragraph*{Quantum theory.}
Consider the usual quantum formalism, for which a state is given by a density matrix $\rho$ on a Hilbert space $\mathcal{H}$.
By decomposing density matrices in an operator basis, we obtain the real vector space $V$ defined above for quantum theory.
For example, there is a well-known representation of the normalized states of a qubit as a linear combination of the Pauli-operators $\sigma_i$:
\begin{align}
 \rho = \frac{1}{2} \, (\id + a \, \sigma_x + b \, \sigma_y + c \, \sigma_z) \qquad a^2+b^2+c^2 \geq 1
\end{align} 
Forming a real vector from the coefficient $a$, $b$, $c$ gives the representation of the qubit state space in $V=\mathbb{R}^3$: this is the Bloch ball. 

Adding a fourth component that indicates normalization gives a representation similar to \eqref{eq:vecrep}.
However, for quantum systems of higher dimension the characterization of the geometrical shape of the state spaces in this representation is still an open problem \cite{Bengtsson11}.

In the usual density matrix representation an effect is a POVM element $E$, which is applied via the trace rule, 
so that the probability of an effect $E$ given the state $\rho$ is given by $\mathrm{Tr}[ E\circ \rho]$.
The unit measure $u$ is given by the identity operator $\mathds{1}$ on $\mathcal{H}$, so that a density matrix $\rho$ is normalized when:
\[
\mathrm{Tr}[ \mathds{1}\circ\rho]=1.
\]
Note that for quantum systems the set of states and the set of effects can be \em identified\em: this is the set of positive operators on $\mathcal{H}$. 
For example, for a qubit the Bloch ball represents both (normalised) states and effects. 
This is an example of `self-duality' in a theory; we shall discuss this further in section \ref{sec:polymodel}.

\paragraph*{Classical probability theory.}  
The state space of a classical system in $\mathbb{R}^d$ is a simplex. 
This is the convex hull of $d+1$ pure states (which can be characterized via a condition on linear independence).
For example, for $d=1$, the classical state space is a geometrically line, which represents a bit.
The extreme points of the line $\omega_0$ and $\omega_1$ are the pure states: these represent the values $0$ or $1$ of the bit  respectively. 
The convex mixtures $p\omega_0+(1-p)\omega_1$ represent states of classical uncertainty about the value of the bit.
Only one measurement outcome is needed to identify the state, e.g.~the probability of obtaining the $0$ value for the bit.
For $d=2$, the simplex is a triangle in $\mathbb{R}^2$, which represents a trit; and so on.
As for a bit, for any $d$ the pure states $\omega_i$ represent mutually exclusive properties of the system. 
For example, if one knows with certainty which number is on top of a die, then one automatically knows that none of the other numbers is on top.
This means that the pure effects then correspond to measurement outcomes that perfectly distinguish $\omega_i$ 
i.e. $e_i(\omega_j) = \delta_{ij}$.

\paragraph*{Boxworld.} This is a popular toy theory in the GPT framework that is neither quantum nor classical, which was first introduced systematically in \cite{Barrett07}.
Boxworld consists of a  class of single systems characterized by the dimension $d \geq 2$ of the state space.
For $d=2$ the normalized state space $\Omega$ is the convex hull of the following pure states:
\begin{align}
 \label{eq:gbitstates}
 \omega_1 = (1,0,1)^T \quad &\omega_2 = (0,1,1)^T \\
 \omega_3 = (-1,0,1)^T \quad &\omega_4 = (0,-1,1)^T,
\end{align}
and so geometrically $\Omega$ is a square.
The set of effects is given by the convex hull of of $\emptyset = (0,0,0)^T$, $u=(0,0,1)^T$  and the following extremal effects:
\begin{align}
 \label{eq:unrestrictedgbiteffects}
 e_1 = \frac{1}{2} \, (1,1,1)^T \quad &e_2 = \frac{1}{2} \, (-1,1,1)^T \\
 e_3 = \frac{1}{2} \, (-1,-1,1)^T \quad &e_4 = \frac{1}{2} \, (1,-1,1)^T
\end{align}
It is straightforward to show that the measurement statistics of the two orthogonal binary measurements $M_1=\{e_1,e_3\}$ and $M_2=\{e_2,e_4\}$ give enough information to identify any state.
Indeed, due to the normalization constraint the measurement statistics of the binary measurements on normalized states is determined by the probabilities $p_1$, $p_2$ for the first outcomes $e_1$, $e_2$. 
The different states give rise to the full range $(p_1,p_2) \in [0,1]^2$ of possible probability distributions, with the probabilities $p_1$ and $p_2$ being independent.
Hence the measurement outcomes $e_1$ and $e_2$ are enough to identify the state of the system, which verifies that the dimension is $d=2$.
Note that unlike orthogonal measurements in quantum theory (such as $\sigma_x$ and $\sigma_y$), there is no uncertainty principle for $M_1$ and $M_2$ for this system \cite{steeg}.
For example, although $e_1$ and $e_4$ belong to orthogonal measurements, we have $e_1(\omega_1)=e_4(\omega_1)=1$. 

Higher dimensional single systems with $d > 2$ in boxworld have $d$ different binary orthogonal measurements and state spaces given by hypercubes.
For the joint states that we shall discuss in section \ref{sec:jointstates}, boxworld allows maximal nonlocal correlations (using the CHSH inequality introduced below). 
These correlations define the \em Popescu-Rohrlich box \em \cite{PRbox}, and they are not realizable by quantum theory.

\section{The no-restriction hypothesis}
\label{sec:norestriction}
We now consider in detail the no-restriction hypothesis, and the consequences of relaxing it.

\subsection{Defining the set of effects}

Effects are restricted to give values in the range of $[0,1]$ when applied to normalized states.
But in the traditional framework of GPTs, the set of effects $E$ is not restricted any further.
That is, the set of effects is exactly the set of all probability-valued linear functionals on the given states.
We will call this relationship between states and effects the \em no-restriction hypothesis\em, in accordance with \cite{Chiribella11}. 
It is satisfied for classical probability theory and quantum theory.

\begin{thm}
 The set of effects under the no-restriction hypothesis is given by
\begin{align}
\label{eq:fulleffectset}
E := V^*_+ \cap (u - V^*_+)
\end{align}
 with the so-called dual cone
\begin{align}
 \label{eq:dualcone}
 V^*_+ &:= \left\{ e \in V^* \, \middle| \, e(\omega) \geq 0 \quad \forall \omega \in \Omega \right\}.
\end{align}
\end{thm}
 
\begin{proof}
The definition of effects as probability-valued linear functionals can be decomposed into two conditions. 

The first condition is that effects have to give non-negative results on every element of the state space.
For arbitrary elements $e \in V^*$ satisfying this condition, the condition is also satisfied by the positive ray $\left\{ \lambda \, e \middle| \lambda \geq 0 \right\}$.
Hence, the set obeying the non-negativity condition is a cone, namely the \em dual cone \em $V^*_+$, defined by \eqref{eq:dualcone}.

The second condition on effects requires them to give results not larger than one, when applied to arbitrary normalized states.
In other words the results have to be one or one minus a positive value, i.e. $e \in u - V^*_+$.

In the standard framework both boundary conditions are saturated.
That is, for a given state space, \em any \em linear functional that gives probability-valued results for all normalized states is included in the theory.
Thus, the set of effects $E$ is $V^*_+ \cap (u - V^*_+)$.
\end{proof} 
The dual of the dual cone is the \em primal cone \em $V_+$, which is generated by unnormalized states, i.e.
\begin{align}
\label{eq:primalcone}
 V_+ &:= \left\{ \lambda \, \omega \, \middle| \, \omega \in \Omega, \lambda \geq 0 \right\} = (V_+^*)^* .
\end{align}
Consequently, if the no-restriction hypothesis holds, then a theory is completely determined by the state space, since the effect set can be derived from the state space.

The purpose of this paper is to develop the framework of GPTs without the no-restriction hypothesis. 
There are two main reasons for doing so.
Firstly,  the necessity of the no-restriction hypothesis is questionable from an operational perspective.
Indeed, considering the physical meaning of states and effects there is no reason to believe that the possible preparation procedures determine possible measurements.
Secondly, this will generalize the GPT framework to cover new scenarios that have not been accessible within the old framework.

\subsection{Relaxing the no-restriction hypothesis}

Let us note the constraints that still apply when the no-restriction hypothesis is removed. 
Clearly, effects still need to give probabilities when applied to any state. 
That is, when allowing violations of the no-restriction hypothesis, the set of probability-valued linear functionals on states in \eqref{eq:fulleffectset} remains an upper bound for possible effects. 
However, in general not all elements in this set need to represent a valid measurement outcome.
Consequently, the set of effects $E$ may actually be given by a \em subset \em of \eqref{eq:fulleffectset}.
This is the crucial new ingredient in the GPT framework that we shall use in subsequent sections.

Furthermore, we have identified the following four consistency conditions that also have to be met: 
\begin{itemize}
\item[i)] The unit measure $u$ needs to be included in the restricted set as it is crucial for the definition of measurements.
\item[ii)] For every effect $e$ included in $E$, the complement effect $\bar{e} = u-e$ needs to be included as well.
We will show in section \ref{sec:linmaps} that including an effect, but not the complement can yield inconsistencies for joint states.
\item[iii)] Coarse graining also provides effects that can be derived from existing ones.
If one does not distinguish between some measurement outcomes that are part of the same measurement, the common probabilities are given by the sum of the individual probabilities.
Due to linearity the corresponding effect describing the coarse graining is given by the sum of the individual effects. 
\item[iv)] Transformations map valid states to valid states. However, for any transformation $T$ on states, there is also an adjoint transformation $T^\dagger$ on effects defined by $e[T(\omega)]=[T^\dagger (e)](\omega)$ for all states and effects. Thus, the effect set has to respect given transformations.
\end{itemize}

Apart from these consistency restrictions, the definition of the effect set $E$ is now an independent part of the specification of the theory. 
In other words, the effect set $E$ does not depend on the state space now, and the dual cone $V^*_+$ is irrelevant for single systems.
However, we will see in section \ref{sec:gentensormax} that we still need it to classify consistent joint states.

Let us now consider how removing the no-restriction hypothesis will be useful.
As shown above, the no-restriction hypothesis connects a set of states and effects via the respective dual-cone. 
Taking a closer look at the dual cone construction in \eqref{eq:dualcone}, it can easily be seen that each extremal point of the primal cone describes a facet of the dual cone and the other way round.
Therefore, arbitrary small changes in the primal cone, can have an enormous impact on the form of the dual cone. 
Consequently, the no-restriction hypothesis makes it extremely difficult to alter a theory in a controlled way. 
However, it has always been a central motivation for the framework of generalized probabilistic theories to find alternatives to quantum theory.

We shall now show in sections \ref{sec:noisyboxworld} and \ref{sec:polymodel} that new models with interesting features can indeed be constructed when accepting violations of the no-restriction hypothesis.
Furthermore, for joint systems, we will see in sections \ref{sec:jointstates} and \ref{sec:genmaxtensor} how consistency conditions are affected.

\section{Theories with intrinsic noise}
\label{sec:noisyboxworld}

The no-restriction hypothesis guarantees that for any pure state $\omega$, there is an effect $e$, with $e\neq u$, such that ${e(\omega)=1}$.
In contrast, removing the no-restriction hypothesis allows for the modeling of systems with intrinsic noise, i.e. systems for which the unit measure is the \em only \em certain outcome for any state.
For example, an isotropic unbiased implementation of noise can be achieved  by restricting the effects to a set where the original extremal effects are replaced by mixtures with $u/2$ (except for $\emptyset$ and $u$ itself).
In order to combine noise and bias one can mix the extremal with another effect instead of $u/2$.

\begin{figure}
 \label{fig:noisyboxworld}
 \begin{tikzpicture}[scale=1.5]
  \draw[thick, dashed, red] (1,1)--(-1,1)--(-1,-1)--(1,-1)--cycle;
  \draw[thick, red, fill=red!50] (.7,.7)--(-.7,.7)--(-.7,-.7)--(.7,-.7)--cycle;  
  \draw[thick, blue, fill=structure!50!white, opacity=.5] (1,0)--(0,1)--(-1,0)--(0,-1)--cycle;
  \draw[red,fill=red] (0,0) circle (1pt);
  \draw[->,thick] (.95,.95)--(.75,.75);
  \draw[->,thick] (-.95,.95)--(-.75,.75);
  \draw[->,thick] (.95,-.95)--(.75,-.75);
  \draw[->,thick] (-.95,-.95)->(-.75,-.75);
  \node at (0,.15) {\color{red} $\{u,\emptyset\}$};
  \node at (1.15,0) {\color{structure} $\omega_1$};
  \node at (0,1.15) {\color{structure} $\omega_2$};
  \node at (-1.15,0) {\color{structure} $\omega_3$};
  \node at (0,-1.15) {\color{structure} $\omega_4$};
  \node at (1.15,1.15) {\color{red} $e_1$};
  \node at (-1.15,1.15) {\color{red} $e_2$};
  \node at (-1.15,-1.15) {\color{red} $e_3$};
  \node at (1.15,-1.15) {\color{red} $e_4$};
  \node at (.85,.65) {\color{red} $e^\lambda_1$};
  \node at (-.85,.65) {\color{red} $e^\lambda_2$};
  \node at (-.85,-.65) {\color{red} $e^\lambda_3$};
  \node at (.85,-.65) {\color{red} $e^\lambda_4$};
  \draw[red,fill=red] (0.7,0.7) circle (1pt);
  \draw[red,fill=red] (-0.7,0.7) circle (1pt);
  \draw[red,fill=red] (0.7,-0.7) circle (1pt);
  \draw[red,fill=red] (-0.7,-0.7) circle (1pt);
  \draw[structure,fill=structure] (1,0) circle (1pt);
  \draw[structure,fill=structure] (0,1) circle (1pt);
  \draw[structure,fill=structure] (-1,0) circle (1pt);
  \draw[structure,fill=structure] (0,-1) circle (1pt);
 \end{tikzpicture}
 \caption{Inclusion of noise into boxworld: State space and effects are both embedded into $\mathbb{R}^3$ and shown from above for illustration. The state space (blue) is given by a square. The effect set is the octahedron spanned by the extremal effects $e_i$, $u$ and $\emptyset$. The noisy theory has a restricted effect set with extremal effects $e^\lambda_i$.} \label{fig:noisyboxworld}
\end{figure}
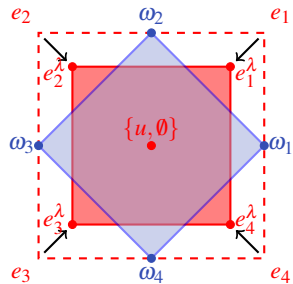

The inclusion of intrinsic noise by a modification of boxworld is illustrated in Fig.~\ref{fig:noisyboxworld}.
The state space of a single system is given by a square. 
In the traditional model the effect set is determined by the no-restriction hypothesis.
A noisy version of boxworld is given by mixing the extremal effects $e_i$ with $u/2$:
\begin{align}
 e_i \mapsto e_i^\lambda = \lambda \, e_i + (1-\lambda) \, \frac{u}{2}
\end{align}
The strength of noise is given by $(1-\lambda)$, i.e. the maximal probability from extremal effects is $\lambda$.

This model is particularly interesting with respect to its potential non-local correlations in joint systems. 
This will be examined in more detail after introducing joint states in section \ref{sec:jointstates}.

\section{Self-dualization procedure}
\label{sec:polymodel}

\begin{figure*}[t]
 \begin{tikzpicture}
  \node (polyorig) {
   \begin{tikzpicture}
    \draw[dotted] (0,0) circle (1.3);
    \draw[thick, red, fill=red!30] (30:1.3)--(90:1.3)--(150:1.3)--(210:1.3)--(270:1.3) --(330:1.3)--cycle; 
    \draw[thick, blue, fill=structure!50!white, opacity=.5] (0:1.3)--(60:1.3)--(120:1.3)--(180:1.3)--(240:1.3)--(300:1.3)--cycle;  
    \draw[red,fill=red] (0,0) circle (1pt);
    \foreach \i/\itext in {60/1,120/2,180/3,240/4,300/5,360/6} \node at (\i:1.45) {\color{structure} $\omega_\itext$};
    \foreach \i/\itext in {30/1,90/2,150/3,210/4,270/5,330/6} \node at (\i:1.45) {\color{red} $e_\itext$};
   \end{tikzpicture}
  };
  \node[right of=polyorig, node distance=4cm] (polyscaled) {
   \begin{tikzpicture}
    \draw[dotted] (0,0) circle (1.3);
    \def\rstat{1.21}
    \def\reff{1.397}  
    \draw[thick, red, fill=red!30] (30:\reff)--(90:\reff)--(150:\reff)--(210:\reff)--(270:\reff) --(330:\reff)--cycle; 
    \draw[thick, blue, fill=structure!50!white, opacity=.5] (0:\rstat)--(60:\rstat)--(120:\rstat)--(180:\rstat)--(240:\rstat)--(300:\rstat)--cycle;  
  
    \draw[red,fill=red] (0,0) circle (1pt);
    \foreach \i/\itext in {60/1,120/2,180/3,240/4,300/5,360/6} \node at (\i:1.45) {\color{structure} $\omega_\itext$};
    \foreach \i/\itext in {30/1,90/2,150/3,210/4,270/5,330/6} \node at (\i:1.6) {\color{red} $e_\itext$};
   \end{tikzpicture}
  };
  \node[right of=polyscaled, node distance=4cm] (polytrunc) {
   \begin{tikzpicture}
    \draw[dotted] (0,0) circle (1.3);
    \def\rstat{1.21}
    \def\reff{1.397}  
    \draw[thick, gray!20, fill=gray!10, opacity=.5] (30:\reff)--(90:\reff)--(150:\reff)--(210:\reff)--(270:\reff) --(330:\reff)--cycle; 
    \draw[thick, red, fill=red!30] (0:\rstat)--(60:\rstat)--(120:\rstat)--(180:\rstat)--(240:\rstat)--(300:\rstat)--cycle;  
    \draw[thick, blue, fill=structure!50!white, opacity=.5] (0:\rstat)--(60:\rstat)--(120:\rstat)--(180:\rstat)--(240:\rstat)--(300:\rstat)--cycle;  
  
    \draw[red,fill=red] (0,0) circle (1pt);
    \foreach \i/\itext/\rad in {60/1/1.45,120/2/1.45,180/3/1.65,240/4/1.45,300/5/1.45,360/6/1.65} \node at (\i:\rad) {{\color{structure} $\omega_\itext$},{\color{red}$e'_\itext$}};
    \foreach \i/\itext in {30/1,90/2,150/3,210/4,270/5,330/6} \node (p\i) at (\i:1.6) {\color{gray} $e_\itext$};
 \node (x1) at (0,0){};
 \node (x2) at (1,1){};
   \end{tikzpicture}
  };
  \draw[->,thick] (1,1) to[bend left] (3,1);
  \draw[->,thick] (5,-1) to[bend right] (7,-1);
 \end{tikzpicture}
\caption{Self-dualization of a hexagon system: The pictures show the statespace (blue) and the intersection of the effect cone (red) that lies in the same plane. In the first step the state cone will be embedded into the effect cone by an equivalence transformation \eqref{eq:equivalentsets}. In the second step the effects not included in the state cone are abandoned.} \label{fig:selfdual} 
\end{figure*}
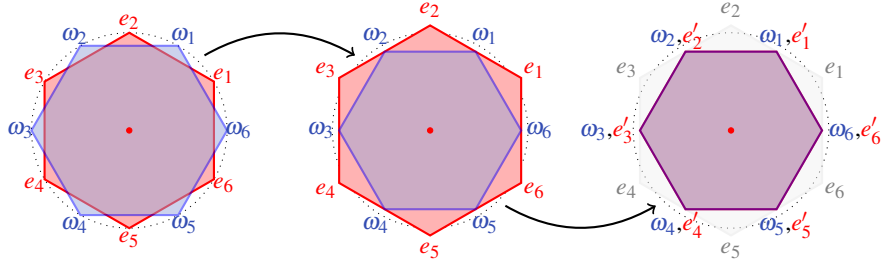

A particular class of systems that has gained a lot of interest recently are so-called \em (strongly) self-dual \em systems \cite{Barnum08,polypaper,Mueller12}. 
These are systems with a particular geometrical structure, shared by both classical probability theory and quantum theory.
For strongly self-dual systems states and effects can be identified with each other and thus be represented by the same mathematical objects.
E.g. in quantum theory both states and effects are represented by positive hermitian operators.

Formally, strong self-duality is given by the following definition.
\begin{dfn} 
\label{def:stronglyselfdual}
 A system is \em strongly self-dual \em iff there exists an isomorphism $\Phi: V^*_+ \mapsto V_+$ giving rise to a corresponding symmetric bilinear form $T$ with $T(e,f)=e[\Phi(f)]=T(f,e)$ and $T(e,e) \geq 0$ for all $e,f \in V^*$.
\end{dfn}

That is, $T$ provides a semi inner product on effects. 
In a similar way for strongly self-dual systems the inverse map $\Phi^{-1}$ leads to a semi inner product on states.

Strong self-duality greatly restricts the class of possible systems.
As we describe below, the property of `bit-symmetry' implies that a system is strongly self-dual \cite{Mueller12}, and there is evidence that non-local correlations of self-dual systems are limited \cite{polypaper}.
In this section we provide a general construction rule to modify any system, such that it resembles the behaviour of strongly self-dual systems.

\begin{thm}\label{thm:selfdualize}
 Any theory in the GPT framework can be modified to resemble strongly self-dual systems respecting Definition \ref{def:stronglyselfdual} with the dual cone $V^*_+$ replaced by a truncated cone $\mathcal{V}^*_+$.
\end{thm}

\begin{proof}
Using our representation, we assume an embedding of effects and states in a common vector space with a scalar product mediating the application of effects on states, as in Eq.~\ref{eq:vecrep}.
We start from an arbitrary theory for which the no-restriction hypothesis holds.
The freedom of linearly transformations $L^T$ from \eqref{eq:equivalentsets} allows us to strictly enlarge the effect cone $V^*_+$, while the corresponding inverse $L^{-1}$ constricts the cone of unnormalized states $V_+$ to be strictly smaller.
Hence, one can always represent the same physical theory, with $V_+$ embedded in $V^*_+$.
We can then define a truncated the effect cone from $\mathcal{V}^*_+ \subseteq V^*_+$, such that 
$\mathcal{V}^*_+$ coincides with the state cone $V_+$. 
Hence we can describe unnormalized effects and states with the same set of vectors.
Consequently, the restriction of effects yields the vector space's scalar product to act as an inner product between states.
This satisfies the definition of strong self-duality with the dual cone $V^*_+$ exchanged for the truncated effect cone $\mathcal{V}^*_+$.
The set of effects is then constructed from $\mathcal{V}^*_+$ by $E = \mathcal{V}^*_+ \cap u-\mathcal{V}^*_+$.
\end{proof}

The connection between self-dualized systems and actual strongly self-dual systems is not only limited to a mere formal resemblance.
In fact, the following example shows that self-dualized systems have features that strongly self-dual systems have when the no-restriction hypothesis is assumed.

\subsection{Example: self-dualized polygons}

Let us illustrate the self-dualization procedure on a set of systems introduced in a previous paper \cite{polypaper}.
It is defined by two-dimensional state spaces with the shape of regular polygons.
While the cases with an odd number of vertices $n$ are strongly self-dual, the even cases are not.

For fixed $n$, let $\Omega$ be the convex hull of $n$ pure states $\{\omega_i\}$, $i=1,...,n$, with
\begin{equation}\label{eq:localpolygons}
 \omega_i = \begin{pmatrix}
             r_n \cos(\frac{2 \pi i}{n})\\
             r_n \sin(\frac{2 \pi i}{n})\\
             1
            \end{pmatrix} \in \mathbb{R}^3 ,
\end{equation}
where $r_n= \sqrt{\sec(\pi/n)}$.

The unit effect is
\begin{equation}
   u = (0,0,1)^T.
\end{equation}

The set $E(\Omega)$ of all possible measurement outcomes will be determined by the no-restriction hypothesis.
In the case of even $n$, $E(\Omega)$ is the convex hull of the zero effect, the unit effect, and $e_1,\ldots, e_n$, with
\begin{equation}
  \label{eff_even}
  e_i = \frac{1}{2} \, \begin{pmatrix}
    r_n \cos(\frac{(2 i-1) \pi}{n})\\
    r_n \sin(\frac{(2 i-1) \pi}{n})\\
    1
  \end{pmatrix}  .
\end{equation}

The odd case yields a different expression for the ray-extremal effects
\begin{equation}
  \label{eff_odd}
  e_i = \frac{1}{1 + {r_n}^2} \, \begin{pmatrix}
    r_n \cos(\frac{2 \pi i}{n})\\
    r_n \sin(\frac{2 \pi i}{n})\\
    1
  \end{pmatrix}.
\end{equation}
As shown in \ref{sec:norestriction} the complement effects $\bar{e_i} = u - e_i$ of ray extremal effects $e_i$ are also extremal in the effect set $E(\Omega)$.
Whereas for even $n$ these happen to coincide with $\bar{e_i} = e_{(i+n/2) \mathrm{mod} \ n}$, for odd $n$ the complement effects form additional extremal points of $E(\Omega)$.
In summary, $E(\Omega)$ is the convex hull of the zero effect, the unit effect $u$, the ray-extremal effects $e_1\ldots,e_n$, and for odd $n$ additionally $\bar{e_1},\ldots,\bar{e_n}$.

In the limit $n \to \infty$ both cases converge to a disc that can be regarded as the 2D subspace of a qubit.
The extremal rays of the dual cone of polygon systems with odd number of vertices, coincide with the scaled extremal states, i.e. these systems are strongly self-dual. 
However, for polygon system with an even number of vertices the primal and dual cones are only isomorphic and can be matched by a rotation of $\frac{\pi}{n}$.
That is, the even polygons are not strongly self-dual in the original models.
We will now self-dualize these even-polygon systems using the procedure described in Theorem \ref{thm:selfdualize}.  

As discussed in section \ref{sec:equivreps} there is always the freedom to apply arbitrary bijective linear maps to all effects and the corresponding inverse map on all states.
We use this to shrink the state space by $r_n \mapsto 1$ to fit in a circumscribed circle of radius one.
Applying the inverse map to effects results in a effect cone with $r_n \mapsto r_n^2$. 
This new effect cone is strictly bigger than the cone of unnormalized states. 
By truncating this effect cone, such that the new extremal effect $e'_i$ are given by
\begin{align}
 e'_i &= \frac{1}{2} \, \left(e_i + e_{(i+1) \mathrm{mod} \ n}\right) = \frac{1}{2} \begin{pmatrix}
 \cos\left(\frac{2 \, \pi i}{n}\right)\\
 \sin\left(\frac{2 \, \pi i}{n}\right)\\
 1
 \end{pmatrix} = \frac{\omega_i}{2},
\end{align}
the primal cone coincide with the new effect cone generated by the restricted effect set.    

Let us demonstrate the self-dualization procedure explicitly,  by using the polygon with $n=4$ (this is the boxworld model).
In the first step the pure states and effects are transformed to the equivalent representation given in \eqref{eq:gbitstates} and \eqref{eq:unrestrictedgbiteffects}.
In this representation the effect cone is completely embedded in the cone of unnormed states. 
The actual self-dualization is then done by exchanging $e_i$ for $e'_i = \omega_i / 2$, shrugging off the effects not included in the primal cone.

For all self-dualized polygon models, another interesting feature emerges for the restricted case.
Namely, there exists a specific pure state $\bar{\omega}$ for each pure state $\omega$, such that they can be perfectly distinguished by an effect $e$ with $e(\omega)=1$ and $e(\bar{\omega})=0$.
Furthermore, each pair of perfectly distinguishable states can be mapped reversibly to any other pair of perfectly distinguishable states.
This feature is known as \em bit symmetry, \em and was shown to only hold for strongly self-dual systems in the traditional framework \cite{Mueller12}.

This demonstrates that the self-dualization procedure can actually reproduce properties thought to be specific for actual strongly self-dual systems.
Note that the mathematical description of actual strongly self-dual systems can be complex.
Using  self-dualized systems might be an alternative that helps to identify new features of strongly self-dual systems,
even if one is not interested in the relaxation of the no-restriction hypothesis.

\subsection{Spekkens's toy theory}

In \cite{Spekkens} Spekkens introduced a toy theory which replicates many features of quantum theory. 
For example, it exhibits a no-cloning theorem and a teleportation protocol.
The theory is not explicitly probabilistic, since outcomes are not explicitly assigned probabilities.
Instead, a graphical calculus is used.
Given a state $\omega$, the outcome $i$ is only specified to  be `possible' or `impossible'.
The Spekkens theory in its original form also has no notion of arbitrary \em convex mixing\em, i.e.~it does not have the property for any pair of states $\omega_1$ and $\omega_2$, there exists a state $p\omega_1+(1-p)\omega_2$ for all probabilities $p\in [0,1]$.

The ability to form convex mixtures is crucial to GPTs, and in particular to its operational motivation.
Fortunately, there is a natural extension of Spekkens theory which is probabilistic and which does allow convex mixing (the probabilistic version of this theory was also introduced previously by Hardy in \cite{Hardy}).
The state space $\Omega$ of a single system is then the octahedron.
In the representation that we have used, the six extremal states (i.e.~the pure states) are just given by  the co-ordinates of the octahedron in $\mathbb{R}^3$, with an extra component for normalization.
\begin{figure}[h!]
\centering
\includegraphics[width=.3 \linewidth]{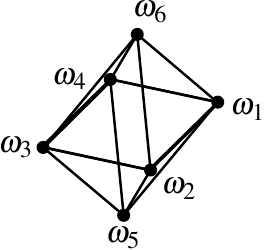}
\caption{The state space of the Spekkens model, with the six pure states $\omega_i$ labelled.}\label{fig:spekkens}
\end{figure}\label{fig:spekkens}
For example, the four extremal states that form the square base of each tetrahedron are identical to the states for boxworld (see Fig.~\ref{fig:spekkens}).
That is, for $i=1,\dots, 4$ the states are:
\begin{equation}\label{eq:localpolygons}
 \omega_i = \begin{pmatrix}
             \cos(\frac{2 \pi i}{4})\\
	    \sin(\frac{2 \pi i}{4})\\
            0 \\
             1
            \end{pmatrix} \in \mathbb{R}^4 ,
\end{equation}
and for $i=5,6$ the states are
\begin{equation}\label{eq:localpolygons}
 \omega_i = \begin{pmatrix}
             0\\
	    0\\
            \pm1 \\
             1
            \end{pmatrix} \in \mathbb{R}^4 ,
\end{equation}

Now, the dual space of an octahedron is the cube. 
However, in the Spekkens theory, the space of effects is \em identical \em to the state space: it is also the octahedron depicted in Fig.~\ref{fig:spekkens}.
Since the octahedron can be obtained by restricting the cube (in the same way that is depicted for the hexagon in Fig.~\ref{fig:selfdual}),we see that the Spekkens theory provides an example of a self-dualized theory. 
In particular, the convex probabilistic version of it is obtained using the self-dualization procedure defined in Theorem \ref{thm:selfdualize}, and as described above for self-dualized polygons.
Indeed, as with boxworld, the restricted effects are given by:
\[
e'_i=\frac{\omega_i}{2}
\]

Hence we see that, at least for single systems, the Spekkens theory can be seen as an \em extension \em of self-dualized boxworld: the state and effect space of the Spekkens theory contain the state and effect space respectively of self-dualized boxworld.
We develop the analysis of joint systems for the Spekkens theory in Section \ref{sec:spekkensjoint}.

We note that the single-system state space is identical to that of \em stabilizer quantum mechanics\em, for which the only allowed states are the eigenstates of the Pauli operators, and the allowed transformations are the Clifford operations.
As discussed in \cite{Spekkens} and further in \cite{Coecke}, the Spekkens theory and stabilizer quantum mechanics differ in the group of reversible transformations that each theory specifies.

\section{Joint systems in the traditional GPT framework}
\label{sec:jointstates}

In the preceding sections we have not distinguished between single systems and joint systems.
That is, our discussion so far (e.g. of self-dualization) has not involved any potential \em subsystem \em structure, whereby a system $C$ can be divided into subsystems $A$ and $B$, with each subsystem having well-defined states and effects.
In the next section we shall consider how relaxing the no-restriction hypothesis affects composite systems. 
Before doing so, in this section we recall the treatment of joint systems in the traditional framework, i.e.~when the no-restriction hypothesis is assumed to hold.

We will restrict the discussion of joint systems to the bipartite case with two subsystems, as the generalization of multipartite systems is straightforward. 
Bipartite joint states are given by elements of the product space
\begin{align} 
\label{eq:productspace}
V^{AB} = V^A \otimes V^B
\end{align}
and joint effects are elements of $V^{AB*} = V^{A*} \otimes V^{B*}$ respectively \footnote{It can be shown that this follows from two conditions on joint states: i) local tomography ii) the no-signalling principle. The no-signalling principle forbids sending information by local operations on a joint state, and will be explained in more detail in section \ref{sec:linmaps}. Local tomography is the identification of joint states by combinations of local measurements.}.

We will represent joint states and joint effects by $n \times m$ matrices, with $n = \dim V^A = \dim V^{A*}$, $m = \dim V^B = \dim V^{B*}$.
As for single systems, the application of effects on states results in the sum of the entry-wise products.
This can be elegantly written as the Hilbert-Schmidt inner product
\begin{align}
 e^{AB}\left(\omega^{AB}\right) &= \Tr \left(e^T \!\!\!\cdot \omega\right) = \sum_{i,j} \epsilon_{ij} \, w_{ij},
\end{align}
where we write $e^T \!\!\!\cdot \omega$ for the matrix product between the transpose of matrix $e$ representing the joint effect $e^{AB}$ and the matrix $\omega$ representing the joint state $\omega^{AB}$. 

To define a composite system for a particular GPT (with specified state and effect spaces for individual systems), we must define the set of joint states $\Omega^{AB}=\{\omega^{AB}\}$, and the set of joint effects $E^{AB}=\{e^{AB}\}$, such that these are consistent with the individual systems.
If the no-restriction hypothesis holds, then, as before, once the set of joint states $\Omega^{AB}$ is defined, the set of effects $E^{AB}$ is determined. 
In this situation we need only consider the definition of $\Omega^{AB}$ in order to specify the behaviour of composite systems.
There is much freedom in defining $\Omega^{AB}$, but there are two boundary cases which we now discuss.

\subsubsection{Lower bound on joint systems}

Consider \em independently prepared \em systems $A$ and $B$ with states $\omega^A \in \Omega^A$, $\omega^B \in \Omega^B$.
Treating the systems \em jointly \em as a composite $AB$, the overall preparation is represented by the product state $\omega^{AB} = \omega^A \otimes \omega^B$, with $\omega^{AB}\in V^{AB}$.
However, just as classical mixtures are allowed for single systems, for joint systems mixtures between product states give valid joint states again.
This corresponds to the ability of experimenters to classically correlate the preparations and measurements of the individual systems, e.g. two experimenters can agree on specific settings.

The set of unnormalized states only containing product states and their mixtures is known as the \emph{minimal tensor product} $A_+ \tensormin B_+$.

\begin{dfn}
\label{def:tradtensormin}
 The \emph{minimal tensor product} is given by
 \begin{align}
  A_+ \tensormin B_+ &:= \left\{\omega^{AB} \in V^{AB} \,\middle|\, \omega^{AB} = \sum_i \lambda_i \, \omega_i^A \otimes \omega_i^B, \right. \\\nonumber
 {} & \qquad\qquad\qquad\qquad \left. \omega_i^A\in A_+, \omega_i^B \in B_+, \lambda_i \geq 0 \right\}.
 \end{align}
 It is the smallest possible set of unnormalized joint states $\omega^{AB}$ that is compatible with given state cones $A_+ \equiv V^A_+$, $B_+ \equiv V^B_+$ of subsystems $A$,$B$.
\end{dfn}

Similar reasoning applies to measurements, and so the set of joint effects is lower-bounded by the convex hull of product effects.
Importantly, this includes the joint unit measure $u^{AB} = u^A \otimes u^B$, which is uniquely defined due to the equivalence principle.
Hence, normalization of joint states $\omega^{AB}$ is represented by the condition $u^{AB}(\omega^{AB})=1$.
This allows us to define the bipartite state space $ \Omega^{AB}_\mathrm{min}$ corresponding to the minimal tensor product:
\begin{align}
 \label{eq:minimalstatespace}
 \Omega^{AB}_\mathrm{min} &:= \left\{\omega^{AB} \in A_+ \otimes_\mathrm{min} B_+ \,\middle|\, u^{AB} \left(\omega^{AB}\right) = 1\right\}\\
 {} &= \left\{ \omega^{AB} \in V^{AB} \,\middle|\, \omega^{AB} = \sum_i p_i \, \omega_i^A \otimes \omega_i^B, \right. \\\nonumber
 {} & \qquad\qquad\qquad\left. \omega_i^A\in\Omega^A, \omega_i^B \in \Omega^B, p_i \geq 0, \sum_i \, p_i = 1 \right\}.
\end{align}

For classical subsystems (i.e.~a simplex), the joint states and effects defined by the minimal tensor product is sufficient to describe joint classical systems.
Theories with non-classical subsystems, however, allow joint states that cannot be interpreted as a mixture of product states, i.e.~entangled states.
The other extreme to the minimal tensor product allows all possible entangled states, as we now show.

\subsubsection{Upper bound on joint systems}

Everything introduced so far is valid independent of the no-restriction hypothesis.
This changes now, as we ask for the maximal sets of joint states and effects consistent with the structure of the single systems.

First, let us focus on the traditional GPT framework with single systems obeying the no-restriction hypothesis.
Given a specific state space the no-restriction hypothesis determines the effects for the single systems.
As argued above, the joint system should at least incorporate product effects and their mixtures. 
Applying such joint effects to any potential joint state $\omega^{AB}$ should give probabilities. 
In particular this implies that the joint states form a subset of the following set of linear elements. 

\begin{dfn}
 The \emph{maximal tensor product} is defined as
\begin{align}
\label{eq:tradtensormax}\nonumber
 A_+ \otimes_\mathrm{max} B_+ &:= \left\{ \omega^{AB} \in V^{AB} \middle| \left(e^A \otimes e^B\right) \!\!\!\left[\omega^{AB}\right] > 0, \right.\\
 {} & \qquad\qquad\qquad\quad\left.\forall e^A \in E^A, e^B \in E^B \right\}\\
\label{eq:tradtensormax2} 
 {} &= \left(A_+^* \tensormin B_+^*\right)^*.
\end{align}
 It is the largest possible set of unnormalized joint states $\omega^{AB}$ that is compatible with given state cones $A_+$, $B_+$ of subsystems $A$, $B$ that respect the no-restriction hypothesis.
\end{dfn}

Note that the second equality arises just by definition of the dual cone \eqref{eq:dualcone}.
Hence, we see that the maximal tensor product for states is given by the maximal set of joint states consistent with the minimal tensor product for effects.
Similarly the maximal tensor product for effects is defined as the maximal set of joint effects consistent with the minimal tensor product for states.
Elements in the maximal tensor product, but not in the minimal tensor product are called \emph{entangled}. 

To summarise our constructions in this section:
the definition of a GPT includes the tensor product, which specifies the composition of subsystems.
The minimal and maximal tensor product are only the extreme cases where the joint state space $\Omega^{AB}$ is chosen as smallest or the biggest set compatible with the state spaces $\Omega^A$, $\Omega^B$ of single systems.
In general, a GPT can be defined to include any set of joint states between those extremes. 

For example, the joint state space in quantum theory lies strictly between the minimal and maximal tensor product.
E.g. the partial transposed of density matrices representing entangled states of two qubits or a qubit and a qutrit are known to give invalid states for the quantum tensor product, because they are not positive on all entangled effects \cite{Horodecki96}.
However, these states give positive results for separable measurements, i.e. they are in the maximal tensor product.
Note that these states should not be misunderstood as part of quantum theory, but form a separate toy theory that omits any entangled measurements. 
Nevertheless, the additional states in the maximal tensor product of local quantum systems are useful for the study of entanglement in standard quantum theory, as they correspond exactly to the set of entanglement witnesses.  

\subsubsection{Joint states as linear maps}
\label{sec:linmaps}

For our generalization of the maximal tensor product, we shall use the following conception of joint states.
Joint states can linearly map effects from one part of the joint system to unnormalized states of the other subsystem.
This can be conveniently shown in the representation of joint states as matrices, since
\begin{align}
 \left(e^A \otimes e^B\right) \!\!\!\left[\omega^{AB}\right] &= \Tr \left[\left(e^A \otimes e^B\right)^T \!\!\!\!\!\!\cdot\omega^{AB}\right] = \left(e^A\right)^T \!\!\!\!\!\cdot \omega^{AB} \!\!\cdot e^B.
\end{align}
Using associativity of the matrix product, we can interpret parts of the expression $\left(e^A\right)^T \cdot \omega^{AB} \cdot e^B$ as `effective' states of the subsystems $A$ and $B$. 
We define these \em conditional states \em as
\begin{align}
 \omega^A_{e^B} & :=\omega^{AB} \!\!\cdot e^B\\
 \omega^B_{e^A} & := \left(e^A\right)^T \!\!\!\!\!\cdot \omega^{AB}
\end{align}
These are unnormalized states for system $A$ and $B$ respectively.
Physically, these can be regarded as `post-measurement` states on one part of the joint system, conditioned on a particular measurement outcome on the other part.
This process of remotely preparing a state by a measurement on the other part of a joint state is usually referred to as `steering' \cite{oppenheim}.
It demonstrates that, when measuring only part of a joint system, the joint state acts as a linear map from effects of one side of the system to unnormalized states of the other part.
It can be shown that the maximal tensor product coincides exactly with all possible linear maps of this form, i.e. it corresponds to all potential joint states that have valid conditional states for non-restricted systems \cite{Barnum08}.
This property will be central for the generalization of the maximal tensor product in the next section.

Conditional states at $A$ are unnormalized: they are weighted with the probability of obtaining the corresponding measurement outcome at $B$.
That is, the probability accounts for the potential ignorance of the outcome for observers at $B$.
Consequently, if one knows the measurement outcome in $B$ the effective description of the state in $A$ is given by the normalized conditional state:
\[
\tilde{\omega}^A_{e^B} = \frac{\omega^A_{e^B}}{p(e^B|\omega^{AB})} = \frac{\omega^A_{e^B}}{u(\omega^A_{e^B})}.
\]
The \em marginal state \em or \em reduced state \em $\omega^A_{u^B}$ gives the description of the effective state on part $A$ of a joint state $\omega^{AB}$.
This is a conditional state with $e^B=u^B$, and is already normalized i.e. $\tilde{\omega}^A_{u^B}=\omega^A_{u^B}$.
  
Note that this formalism still applies if the parts of the system are space-like separated, i.e. if there is no causal relationship between the measurement on the system $B$ and  the system $A$.
However, the no-signaling principle states that steering cannot be used to transmit information, i.e. it does not allow for communication faster than the speed of light. 
The relationship between steering and the no-signaling principle is shown by the following theorem.
First, we call a set of effects $\{e^A_i\}_i$, for any system $A$, a \em perfect measurement \em if
\[
\sum_i e^A_i = u^A.
\]
An \em imperfect measurement \em is a set of effects $\{e^A_i\}_i$ that is not a perfect measurement. 

\begin{thm}
Assuming the no-signalling principle, steering implies that all measurements are perfect measurements. 
\end{thm}

\begin{proof}
Consider two observers in part $A$ and $B$ respectively sharing a joint state $\omega^{AB}$. 
The observer in $B$ performs a measurement on his part and gets some measurement outcome $e^B_j$. 
Knowing the outcome the description of the system in $A$ from his point of view is given by the normalized conditional state $\tilde{\omega}^A_{e_j^B}$. 
The other observer `knows' only the coarse graining of the different measurement outcomes.
I.e. from his point of view the state in $A$ is an ensemble of possible `post-measurement states' $\{\tilde{\omega}^A_{e_i^B}\}$.

Remember that the equivalence principle gives a one-to-one correspondence of states and specific measurement statistics.
Consequently, no-signaling requires the state in $A$ after the measurement on $B$ to be identical to the original marginal state $\omega^A_{u^B}$ in order to prevent information transfer, i.e.
\begin{align}
 \sum_i p_i \, \tilde{\omega}^A_{e_i^B} = \sum_i \omega^A_{e_i^B} = \omega^A_{\sum_i e_i^B} = \omega^A_{u^B},
\end{align}
where we used the definition of the normalized conditional state and the linearity of effects.

Since the coarse grained conditional state needs to be equal to the marginal state for any joint state
\begin{align}
 \sum_i e_i^B=u^B.
\end{align}
\end{proof}

We will use the interpretation of the maximal tensor product as the set of all positive linear maps to generalize it for systems violating the no-restriction hypothesis.

\section{The generalized maximal tensor product}\label{sec:genmaxtensor}

As we have discussed, by removing the no-restriction hypothesis, the definition of a physical system now needs a specification of \em both \em the state space and the effect set.
That is, the set of allowed states and the set of allowed effects can be chosen independently---except for the constraints discussed in section \ref{sec:norestriction}.
Let us now consider the specification of joint systems when the no-restriction hypothesis is removed.

The definition of the minimal tensor product $A_{+}\tensormin B_{+}$ makes no reference to the effect sets $E^A$ and $E^B$.
I.e. it is constructed by products and their convex combinations.
Therefore the minimal tensor product can be defined without assuming the no-restriction hypothesis, and hence carries over to our more general situation. 
Indeed, everything that we have introduced for joint systems so far is valid independently of the no-restriction hypothesis --- with one exception.

The exception is the maximal tensor product.
As before, we expect the maximal tensor product to comprise all joint states that are compatible with the given subsystems.
Compatibility can be broken down to two requirements: i) non-negative results on local effects ii) valid conditional states.
For non-restricted systems both requirements are equivalent, as the no-restriction hypothesis implies consistent mappings (i.e. valid conditional states) if and only if local effects give non-negative results on joint states. 
Now, for the general case (i.e. without the no-restriction hypothesis), valid conditional states still guarantees non-negativity on local effects.
However, the implication in the other direction is no longer secured.

For example, consider attempting to use the same construction as before, i.e. we start with the minimal tensor product of effects and determine all elements of the joint system that give positive results.
The resulting elements do not depend on the state spaces of the single systems \em at all\em, since the effects are decoupled from the state space due to the abandoned no-restriction hypothesis.
Hence the resulting joint states are not forced to be consistent with the subsystems: we give an example of such a failure of consistency below.

\subsection{Failure of the traditional maximal tensor product}

Before generalizing the maximal tensor product we will show that the traditional construction rules fail for restricted systems. 

The traditional maximal tensor product $A_+ \tensormax B_+$ is given by the dual of the set of separable effects.
For restricted systems this yields two different variants. 
Equation \eqref{eq:tradtensormax} seems to suggest a construction based on the restricted effects, whereas \eqref{eq:tradtensormax2} utilizes the subsystems' dual cones, which are generated by the potential set of unrestricted effects.
We show that neither choice gives the set of all joint states consistent with restricted subsystems.

The first variant is constructed as follows.
Consider the restricted effects $E^A$ of a subsystem $A$ with an effect cone $E^A_+ := \{\lambda \, e^A \,|\, e^A \in E^A, \lambda \geq 0\}$.
Following equation \eqref{eq:primalcone} we can construct a virtual, non-restricted system $\mathcal{A}$ with the state cone given by
\begin{align}
\label{eq:virtualcones}
 \mathcal{A}_+ &:= \left\{ \omega^A \in V^A \, \middle| \, e^A(\omega^A) \geq 0 \quad \forall e^A \in E^A_+ \right\} \supseteq A_+ \\
 \Rightarrow \mathcal{A}^*_+ &= E^A_+.
\end{align}
I.e. the virtual system extends the unnormalized states, such that the no-restriction hypothesis is satisfied.
Thus, the potential joint states from \eqref{eq:tradtensormax}, correspond actually to the traditional maximal tensor product $\mathcal{A}_+ \tensormax \mathcal{B}_+$ of the virtual systems $\mathcal{A}$, $\mathcal{B}$.

Recall that the interpretation of joint states as positive linear maps, $\mathcal{A}_+ \tensormax \mathcal{B}_+$ is exactly the set of all maps from the restricted effect cones $E^A_+$ ($E^B_+$) to the unnormalized virtual states $\mathcal{B}_+$ ($\mathcal{A}_+$) on the other side of the bipartite system.
In other words, this construction includes joint states that allow the preparation of states in the subsystems not limited to the initial definition of the state spaces $\Omega^A$, $\Omega^B$, but to those of the virtual systems instead. 

For example in a bipartite system of self-dualized boxworld with extremal states according to \eqref{eq:gbitstates} and restricted extremal effects $e'_i = \omega_i / 2$ the potential joint state
\begin{align}\label{eq:counterexample}
 \omega^{AB} &= \begin{pmatrix}
  1 & -1 & 0\\
  1 & 1 & 0\\
  0 & 0 & 1  
 \end{pmatrix} \in \Omega^{\mathcal{AB}}_\mathrm{max}
\end{align} 
gives positive values on any pair of restricted effects.
However, some conditional states are not valid for the actual system $A$, e.g. $\tilde{\omega}^A_{e'_1} = (-1,1,1)^T \notin \Omega^A$.

The second variant of the traditional maximal tensor product is based on the dual cones $A_+$, $B_+$ according to \eqref{eq:tradtensormax2}.
The resulting joint states are also consistent with the restricted effects, since the latter is included in the set of all of effects.
However, this construction omits joint states which are consistent \em only \em with the restricted effects.
For example, for self-dualized boxworld the identity matrix would not be included, although it has valid conditional states and gives positive results on any pair of effects.

\subsection{Construction of the generalized maximal tensor product}
\label{sec:gentensormax}

As shown above, the traditional construction rules for the maximal tensor product lead to inconsistencies when applied to theories not obeying the no-restriction hypothesis.
In this section we shall construct a \em generalized \em maximal tensor product $A_+ \gtensormax B_+$: this will give the maximal set of joint states that is consistent with general subsystems, irrespective of whether the no-restriction hypothesis is assumed to hold.
In other words, the generalized maximal tensor product contains all bipartite states whose conditional (i.e. also marginal) states are elements of the original state spaces.  

\begin{figure*}
\label{fig:gentensormax}
 \begin{tikzpicture}
  \node(jointstate1) {
   \begin{tikzpicture}
    \draw[fill=structure!50!white] (-.4,.75) [rounded corners=10pt]-- (.3,.75) -- (.3,-.75)  [sharp corners]-- (-.4,-.75) -- cycle;
    \node at (-0.05,0) {$\omega^{AB}$};
    \draw[fill=red!50] (-1,.2) [rounded corners=7pt]-- (-1.5,.2) -- (-1.5,.75) [sharp corners]-- (-1,.75) --cycle;
    \node at (-1.25,.5) {$e^A$};
    \draw (-.4,.5) -- (-1,.5) (-.4,-.5)--(-1,-.5);
   \end{tikzpicture}  
  };
  
  \node(conditional1) [below=of jointstate1, node distance=1cm]
  {
   \begin{tikzpicture}
    \draw[fill=structure!50!white] (-.3,.3) [rounded corners=8pt]-- (.4,.3) -- (.4,-.3) [sharp corners]-- (-.3,-.3) --cycle;     
    \draw (-.3,0)--(-.8,0);    
    \node at (0,.3) {$\tilde{\omega}^B_{e^A}$};
    \node at (1.2,.35) {\Large{$\in \Omega^B$}};
   \end{tikzpicture}
  };
  
  \node(ABdirection) [right=of jointstate1, node distance=2cm, yshift=-1cm]{
   \begin{tikzpicture}
    \draw[fill=structure!50!white, opacity=.5] (0,0) circle (50pt);
    \node at (-1.1,1.1) {$\left(E^A_+ \tensormin B^*_+ \right)^*$};
   \end{tikzpicture}
  };

  \node(BAdirection) [right=of ABdirection, xshift=-3.25cm, yshift=-1cm]{
   \begin{tikzpicture}
    \draw[fill=structure!50!white, opacity=.5] (0,0) circle (50pt);
    \node at (-1,-1) {$\left(A^*_+ \tensormin E^B_+ \right)^*$};
   \end{tikzpicture}
  };
  
  \node(jointstate2) [right=of BAdirection, yshift=2cm] {
   \begin{tikzpicture}
    \draw[fill=structure!50!white] (-.4,.75) [rounded corners=10pt]-- (.3,.75) -- (.3,-.75)  [sharp corners]-- (-.4,-.75) -- cycle;
    \node at (-0.45,0) {$\omega^{AB}$};
    \draw[fill=red!50] (-1,-.2) [rounded corners=7pt]-- (-1.5,-.2) -- (-1.5,-.75) [sharp corners]-- (-1,-.75) --cycle;
    \node at (-1.5,-.5) {$e^B$};
    \draw (-.4,.5) -- (-1,.5) (-.4,-.5)--(-1,-.5);
   \end{tikzpicture}  
  };
  
  \node(conditional2) [below=of jointstate2, node distance=1cm]
  {
   \begin{tikzpicture}
    \draw[fill=structure!50!white] (-.3,.3) [rounded corners=8pt]-- (.4,.3) -- (.4,-.3) [sharp corners]-- (-.3,-.3) --cycle;     
    \draw (-.3,0)--(-.8,0);    
    \node at (0,.3) {$\tilde{\omega}^A_{e^B}$};
    \node at (1.2,.4) {\Large{$\in \Omega^A$}};
   \end{tikzpicture}
  };

  \draw[<-,very thick] (3.75,0.5) to[bend right] (1,.7);
  \draw[<-,very thick] (6,-2.5) to[out=90,in=160] (9.25,.7);
 
  \node(gentensormax) at (4.65,-1.5) {{$A_+ \gtensormax B_+$}};

 \draw (-1.5,1.25) rectangle (1.5,-3);
 \draw (7.75,1.25) rectangle (10.75,-3);
 \node[rotate=90] at (9.25,-1.5){\Large{$\Leftrightarrow$}};
 \node[rotate=90] at (0,-1.5){\Large{$\Leftrightarrow$}};

 \end{tikzpicture}

\caption{Illustration of the construction of the generalized maximal tensor product}
\end{figure*}
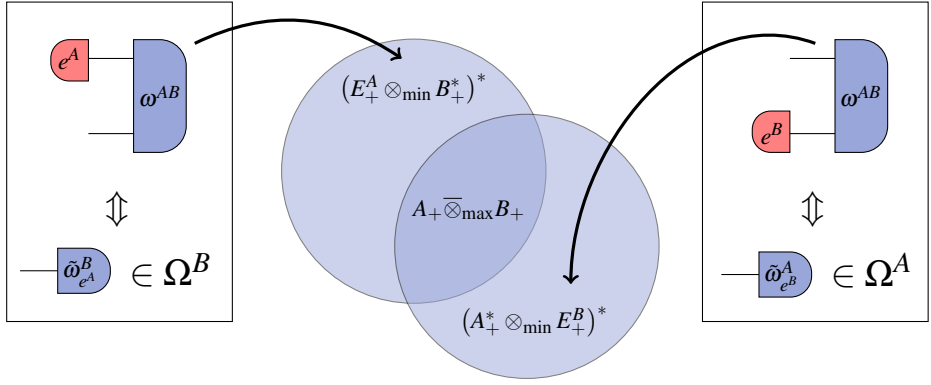

\begin{dfn}
 The \em generalized maximal tensor product \em of systems $A$, $B$ with primal cones $A_+$, $B_+$, dual cones $A^*_+$, $B^*_+$ and effect cones $E^A_+$, $E^B_+$ is given by
\begin{align}
\label{eq:gentensormax}
 A_+ \gtensormax B_+ &:= \left(E^A_+ \tensormin B_+^*\right)^* \!\!\!\cap \left(A_+^* \tensormin E^B_+\right)^* \\
 {} &= \left(E^A_+ \tensormin B_+^* \cup A_+^* \tensormin E^B_+\right)^*. \nonumber
\end{align}
\end{dfn}

For the saturated case, dual cones and effect cones are identical, and we recover the usual maximal tensor product as follows.

\begin{prop}
Suppose that $E^A_{+}=A^*_+$ and $E^B_{+}=B^*_+$.
Then $A_+ \gtensormax B_+=A_+ \tensormax B_+$.
\end{prop}
\begin{proof}
Under the assumptions, Eq.~\ref{eq:gentensormax} becomes
\begin{align*}
A_+ \gtensormax B_+ &= \left(A^*_+ \tensormin B_+^*\right)^*\\
& = A_+ \tensormax B_+
\end{align*}
using the definition of the maximal tensor product in \eqref{eq:tradtensormax}.
\end{proof}
Hence our construction is indeed a generalization of the existing definition of the maximal tensor product.
It determines all joint states consistent with general subsystems regardless whether the no-restriction hypothesis holds or not.
I.e. all joint states with valid conditional states are included, as shown in the following theorem.

\begin{thm}\label{thm:maxtensorprod}
Let $\omega^{AB}\in V^{AB}$.
Then $\omega^{AB}\in A_{+}\gtensormax B_{+}$ iff
$\omega^{AB}$ has well-defined conditional states:
\[
\tilde{\omega}^A_{e^B}\in \Omega^{A} \; \textrm{and} \; \tilde{\omega}^B_{e^A}\in \Omega^{B}
\]
for all $e^A\in E^A$ and $e^B\in E^B$. 
\end{thm}

\begin{proof}

We shall show that 
\begin{align}\label{eq:thmcond1}
{\omega}^A_{e^B}\in A_+ \textrm{  iff  } \omega^{AB}\in\left( E^A_+ \tensormin B_+^*\right)^*
\end{align}
and that 
\begin{align}\label{eq:thmcond2}
{\omega}^B_{e^A}\in B_+ \textrm{  iff  } \omega^{AB}\in\left( A_+^* \tensormin  E^B_+\right)^*.
\end{align}
Since $A_+ \gtensormax B_+$ is defined as the intersection of sets of linear maps $\left( E^A_+ \tensormin B_+^*\right)^*$ and $\left( A_+^* \tensormin  E^B_+\right)^*$, this will establish the thesis.

First we show the $A\rightarrow B$ direction, i.e.~\eqref{eq:thmcond1}, which is the statement that $\left(E^A_+ \tensormin B_+^*\right)^*$ is the set of all and only those joint states $\omega^{AB}$ such that each $\omega^{AB}$ defines a map from effects $e^A \in E^A_+$ on system $A$ to valid unnormalized states $\omega^B \in B_+$.

Recall that for non-restricted systems the traditional maximal tensor product is already known to give all positive linear maps from the effect cone of one part of the system to the state cone of the other part for both directions.
We now show that $\left( E^A_+ \tensormin B_+^*\right)^*$ can be interpreted as the traditional maximal tensor product $\mathcal{A}_+ \tensormax \mathcal{B}_+$ of two virtual systems $\mathcal{A}$ and $\mathcal{B}$ that obey the no-restriction hypothesis.
$\mathcal{A}$ is the virtual system that has already been introduced in \eqref{eq:virtualcones}.
It has an extended virtual state cone $\mathcal{A}_+$, since $(E^A_+)^*\subset \mathcal{A}_+$.
However, the actual effect cone is kept, as it coincides with the dual cone $E^A_+ = \mathcal{A}^*_+$ characterizing unnormalized effects of the non-restricted systems. 
The opposite situation applies to $\mathcal{B}$.
Here, the effect set $E_+^B$ is extended to the dual cone $B^*_+$, so that $E_+^B\subset B^*_+$, where $B^*_+$ representing the full set of potential unnormalized effects.
However, the original state cone $B_+ = \mathcal{B}_+$ is kept.
With these conventions 
\[
\left( E^A_+ \tensormin B_+^*\right)^* = \mathcal{A}_+ \tensormax \mathcal{B}_+ 
\]
follows directly from the definition of the traditional tensor product in \eqref{eq:tradtensormax2}. 
Hence for the $A \to B$ direction, this means that $\left( E^A_+ \tensormin B_+^*\right)^*$ contains all positive linear maps from the restricted effects in $A$ to allowed states in $B$.
That is, $\omega^{AB}\in\left( E^A_+ \tensormin  B^*_+\right)^*$ is a sufficient and necessary condition for ${\omega}^B_{e^A}\in B_+$ which proves \eqref{eq:thmcond1}.
However, for the traditional maximal tensor product the same joint states also coincide with the positive maps in the other direction $\mathcal{B} \to \mathcal{A}$, potentially including invalid mappings.

By swapping the roles of $A$ and $B$ in the above argument, we similarly obtain that the set $\left( A_+^* \tensormin E^B_+ \right)^*$ includes all linear maps that are consistent for the $B \to A$ direction,
but also those which lead to inconsistencies in $A\to B$ opposite direction.
Hence we obtain \eqref{eq:thmcond2}.
\end{proof}

Theorem \ref{thm:maxtensorprod} shows that the generalized maximal tensor product includes only those joint states that are consistent in both directions (i.e. the intersection of the sets $\left( E^A_+ \tensormin B_+^*\right)^*$ and $\left( A_+^* \tensormin E^B_+ \right)^*$).
Note that the if $\omega^{AB}$ has well-defined conditional states, then in particular it is  is locally positive:
\[
\left(e^A \otimes e^B\right) \left[\omega^{AB}\right] \geq 0
\]
which provides a useful necessary condition that joint states must satisfy.

In the traditional GPT framework the choices of tensor products for states and effects are not independent, as the no-restriction hypothesis does not only apply to single systems, but to the joint system as well.
Having the minimal tensor product for joint states (effects) does in fact constitute the maximal tensor product for the set of joint effects (states).
This restriction seems inappropriate given that arbitrary single systems can actually be emulated by classical systems with constrained measurements \cite{Holevo82}, whereas entanglement is a strictly non-classical feature.

In our modified framework that is also valid for systems violating the no-restriction hypothesis, this is no longer the case. 
We have seen that we can generalize the maximal tensor product, but nevertheless we are not forced to use this for states when we choose the minimum tensor product for effects and the other way round.

\subsection{Examples of joint systems}
\label{sec:jointstatesexamples}

To give some specific examples for the generalized maximal tensor product, we have calculated it for the toy theories introduced in sections \ref{sec:noisyboxworld} and \ref{sec:polymodel} using the double description method \cite{doubledescription}.

\subsubsection{Noisy boxworld}

In the original unrestricted version of boxworld joint systems are given by the maximal tensor product, including the $16$ extremal product states and $8$ pure entangled joint states $\Phi^{AB} = \frac{1}{2} (\omega_1 \otimes \omega_2 - \omega_2 \otimes \omega_2 + \omega_2 \otimes \omega_3 + \omega_3 \otimes \omega_1)$ and respectively the states transformed by local symmetries.

These entangled extremals can be interpreted as a maximally entangled state of two such systems, as they form a isomorphic map and have totally mixed reduced states.
They correspond to a rotation of $\frac{\pi}{4}$ and the local symmetries of the state spaces.

This theory has become very popular as it shows nonlocal correlations beyond those possible in quantum theory, when choosing between two possible binary measurements at each side of the bipartite systems.
Let us denote the two measurements $\{M^A_x\}$ and $\{M^B _y\}$ for each of the systems $A$ and $B$ respectively: we index the measurements at each system with $x,y\in\{0,1\}$.
Each measurement has binary outcomes, labelled with $a,b\in\{0,1\}$ for systems $A$ and $B$ respectively.
For example, the $x=0$ measurement on system $A$ consists of a pair of effects $M^A_x=\{e_0,e_1\}$ satisfying $e_0+e_1=u$; similarly for the $x=1$ measurement on system $A$, and $y\in\{0,1\}$ measurements on system $B$.
This leads to a bipartite conditional probability distribution 
\begin{align}\label{eq:bipcondprob}
P(a,b|x,y):=(e_a\otimes e_b)[\omega^{AB}]
\end{align}
We define the correlation 
\[
C_{xy}:=P(a=b|x,y)-P(a \neq b|x,y).
\]
To introduce the Clauser-Horne-Shimony-Holt (CHSH) inequality for demonstrating nonlocality, we introduce the parameter
\[
S:=|C_{00}+C_{01}+C_{10}-C_{11}|,
\]
For classical systems it is upper bounded by the CHSH inequality \cite{CHSH}
\[
\label{eq:CHSHineq}
S^{\mathrm{C}} \leq 2,
\]
whereas for quantum theory it must satisfy $S^\mathrm{Q} \leq 2 \, \sqrt{2}$ \cite{Tsirelson}.
However, local measurements on the maximally entangled state $\Phi$ in boxworld can produce correlations which reach the algebraic maximum $S^\mathrm{max}=4$, i.e. the theory allows the post-quantum correlations known as PR boxes \cite{PRbox}. 

For the noisy version of boxworld introduced in section \ref{sec:noisyboxworld} there is still a notion of a maximally entangled state in the generalized maximal tensor product, namely
\begin{align}
 \Phi^\lambda &= \omega_{\mathrm{ent},1}^{AB} = \mathrm{diag}(\frac{1}{\lambda}, \frac{1}{\lambda}, 1) \cdot \Phi^{1},
\end{align}
i.e. the original maximally entangled state $\Phi^{1}$ combined with a mapping of the effects on one side of the system to the original unrestricted set.
Note that this map does not undo the restriction of effects completely.
The reversion only happens to occur in this particular case when mapping to states of the other part. 
On the other part, however, only restricted effects can be applied to.
Consequently, the correlations possible with restricted systems will be different to those possible in unrestricted systems.

Furthermore, constructing the generalized maximal tensor product it turns out there are 4 different classes of new pure joint states that are entangled but not maximally entangled.
These are representatives of each class
\begin{align}
 \nonumber
 \omega_{\mathrm{ent},2}^{AB} &= -\alpha \, \omega_2 \otimes \omega_2 + \beta \, \omega_2 \otimes \omega_4 + \beta \, \omega_4 \otimes \omega_2 -\alpha \,  \omega_4 \otimes \omega_4\\\nonumber 
 \omega_{\mathrm{ent},3}^{AB} &= -\alpha \, \omega_2 \otimes \omega_2 + \beta \, \omega_2 \otimes \omega_3 + \beta \, \omega_4 \otimes \omega_2 -\alpha \,  \omega_4 \otimes \omega_3\\\nonumber
 \omega_{\mathrm{ent},4}^{AB} &= -\alpha \, \omega_2 \otimes \omega_2 + \beta \, \omega_2 \otimes \omega_4 + \beta \, \omega_3 \otimes \omega_2 -\alpha \,  \omega_3 \otimes \omega_4\\\nonumber
 \omega_{\mathrm{ent},5}^{AB} &= -\alpha \, \omega_3 \otimes \omega_2 + \beta \, \omega_4 \otimes \omega_1 + \beta \, \omega_4 \otimes \omega_2 -\alpha \,  \omega_4 \otimes \omega_3\\
 \text{with } \alpha &= \frac{1-\lambda}{4 \, \lambda}, \beta = \frac{1+\lambda}{4 \, \lambda},
\end{align} 
where the other elements of the class only differ by the local symmetries.

In conclusion the generalized maximal tensor product is spanned by $96$ pure states.
Namely, it consists of $16$ local pure states, $8$ pure entangled states of class $\omega_{\mathrm{ent},1}^{AB}$, $8$ of class $\omega_{\mathrm{ent},2}^{AB}$, $16$ of class $\omega_{\mathrm{ent},3}^{AB}$, $16$ of class $\omega_{\mathrm{ent},4}^{AB}$ and $32$ states of class $\omega_{\mathrm{ent},5}^{AB}$.

Considering local measurements on one instance of any of the nonlocal extremal states the maximal CHSH violation $S^\lambda$ as a function of the parameter $\lambda$ of the restricted model can be shown to be $4 \, \lambda^2$.
Note that this bound is only guaranteed for the correlations that occur from direct measurements.
However, it is known that wiring the measurements on multiple joint states via classical post-processing, might give rise to a distillation of correlations beyond for some values of $\lambda$ \cite{closedunderwirings}.

\subsubsection{Self-dualized polygons}

Interestingly, not only boxworld but all bipartite polygon systems allow a joint state with features known from the maximally entangled state of ordinary quantum theory.
Namely, the linear maps corresponding to these states are given by isomorphisms of the dual and primal cones with maximally mixed reduced states.
The $2 \, n$ different maximally entangled states correspond to the elements of the dihedral group.
For even $n$, the maximally entangled states include an additional rotation of $\pi / n$ mapping the dual cone of one part to the primal cone of the other part.
It was shown that non-local correlations based on two binary local measurements on the maximally entangled states at each side show correlations strictly weaker than quantum correlations for the odd case, whereas the unrestricted even case shows correlations as strong as those of quantum theory or stronger \cite{polypaper}.

Replacing the original polygon systems with even $n$ by their self-dualized versions, the maximally entangled states lose the additional rotation as the new effect cone and the state cone coincide.
Note, that the self-dualized single systems become subtheories of the theory given in the limit $n \to \infty$, i.e. the quantum case, as both states and effects form strict subsets.
Thus, the correlations on the maximally entangled state form a strict subset of those in quantum theory, in contrast to the unrestricted case which allows post-quantum correlations.
Even though the restricted polygons are not genuine strongly self-dual but only self-dualized, this is consistent with the conjecture in \cite{polypaper}, that strong self-duality limits correlations.

For self-dualized boxworld the generalized maximal tensor product is given by the $16$ local pure states, the $8$ states $\omega_{\mathrm{ent},1}^{AB}$ representing the identity and symmetry mappings as well as a class of $64$ pure entangled states $\omega_{\mathrm{ent},2}^{AB} = 1/4 (-\omega_1 \otimes \omega_1 + \omega_1 \otimes \omega_3 + 2 \, \omega_2 \otimes \omega_4 + \omega_3 \otimes \omega_1 - \omega_3 \otimes \omega_3 + 2 \, \omega_4 \otimes \omega_2)$.

Unfortunately, using the double description method, we were not able to characterize all extremals of the generalized maximal tensor product for polygon systems with a higher number of vertices.

\subsubsection{Spekkens's toy theory}\label{sec:spekkensjoint}

The Spekkens theory that we introduced earlier is a local theory, meaning that (in the probabilistic version) it cannot violate any Bell inequalities.
However, as discussed in \cite{Spekkens}, the Spekkens theory has entangled states.
This raises the question of why the Spekkens theory does not exhibit bipartite nonlocality. 
In contrast, a classical theory, i.e. a simplex, is local but it does not have entangled states.
One could then ask, given that the Spekkens theory has entangled states, but is local, what must be \em added \em to the definition of the theory to make it nonlocal?

In our framework, the answer to this question can be clearly understood in terms of the geometry of the state space.
First, recall that the state space $\Omega^{A}$ of a single system in the Spekkens theory is an octahedron, and the effect space $E^A$ is identically the same, i.e. $E^A$ not the full dual space.
Consider a pair of single systems $A$ and $B$ in the Spekkens theory.
Since the effect space $E^A$ is not the full dual space $A^*_+$, we must use the generalized tensor product $\Omega^{AB}=A_+\gtensormax B_+$ to define the bipartite states.
Then consider the following bipartite state:
\begin{align}\label{eq:entstate}
 \omega^{AB} &= \begin{pmatrix}
  0 & 0 & 0 & 0 \\
 0 & -\frac{1}{2} & -\frac{1}{2} &  0 \\
  0 & -\frac{1}{2} & \frac{1}{2} & 0 \\
    0 & 0 & 0 & 1
 \end{pmatrix}
\end{align} 
It is straightforward to verify that $\omega^{AB}$ leads to well-defined conditional states for system $B$ for all effects $e^A\in E^A$, i.e.:
\[
\tilde{\omega}^B_{e^A}\in\Omega^{B}
\]
and correspondingly for conditional states for system $A$ when using effects on system $B$.
In particular, it is also easily checked that $(e^A\otimes e^B)[\omega^{AB}]\geq 0$ for any pair of effects $e^A$ and $e^B$.
Hence by Theorem \ref{thm:maxtensorprod}, this shows that $\omega^{AB}$ is in the generalized tensor product $A_+\gtensormax B_+$ for the Spekkens theory.

Now, since the Spekkens theory is local, the CHSH inequality \eqref{eq:CHSHineq} is satisfied for any choice of measurements $M_x$ and $M_y$ on the state $\omega^{AB}$, or any other bipartite state.
However, let us consider the unrestricted effect space $A^*_+$ from which the restricted space $E^A$ for the Spekkens theory was derived.
The unrestricted effect space of the octahedron is the cube.
We can represent the normalised extremal effects as the vertices of a cube:
\[
e_i = 
\frac{1}{2}
 \begin{pmatrix}
 \pm1 \\
 \pm1\\
 \pm1\\
 1 
 \end{pmatrix}
\]
Now, suppose that we use the cube to be the effect space for the octahedron, i.e.~we use the full dual space.
It is easily shown that the state $\omega^{AB}$ defined in Eq.~\ref{eq:entstate} is again in the generalized maximal tensor product $A_+\gtensormax B_+$.
However, we can now provide measurements which violate the CHSH inequality. 
In particular, consider two measurements for Alice given by $M^A_0=\{e_0,u-e_0\}$ and  $M^A_1=\{e_1,u-e_1\}$ where:
\[
e_0=
\frac{1}{2}
\begin{pmatrix}
1 \\
1 \\
-1 \\
1 
\end{pmatrix}
,
e_1=
\frac{1}{2}
\begin{pmatrix}
-1 \\
-1 \\
-1 \\
1 
\end{pmatrix}
\]
and two measurements for Bob given by $M^B_0=\{e_0,u-e_0\}$ and  $M^B_1=\{e_2,u-e_2\}$, where:
\[
e_0=
\frac{1}{2}
\begin{pmatrix}
1 \\
1 \\
-1 \\
1 
\end{pmatrix}
,
e_2=
\frac{1}{2}
\begin{pmatrix}
-1 \\
-1 \\
1 \\
1 
\end{pmatrix}
\]
By using these choices of measurements in Eq.~\ref{eq:bipcondprob} and the following equations, we obtain the value of the CHSH parameter: this is $S=4$.
This is the value attained by PR boxes, and hence using the full effect space essentially yields the same nonlocality as boxworld. 

We therefore see that the Spekkens theory can be embedded into a nonlocal theory by embedding the effect space of single system into the full dual cone.
Moreover, we see completing the Spekkens theory in this way yields boxworld. 
This provides a new understanding of why the Spekkens theory is local:  the measurements are too restricted.

\section{Conclusions}

We have extended the framework of generalized probabilistic theories.
Given an arbitrary state space the traditional framework determines the possible measurement outcomes as corresponding to the complete set of probability valued linear functionals on states.
In contrast to the traditional framework, our generalization allows the set of states and and the set of effects to be defined separately.
As a result the upper bound for the set of joint states, known as the maximal tensor product, is no longer valid in its traditional form, but has to be replaced by a generalized version.

As an application for restricted models, we provided a self-dualization procedure that alters any theory by restricting the set of effects, such that states and the restricted effects are similarly related as states and unrestricted effects in strongly self-dual systems.
We introduce specific examples for which the self-dualization does not only give a formal resemblance but reproduces a phenomenon called \emph{bit symmetry} shown to only hold for strongly self-dual systems in the traditional framework \cite{Mueller12}.
Furthermore, these self-dualized models show quantum correlations, whereas the original models have correlations that are stronger than quantum correlations. 
In particular, the correlations of boxworld---a theory known to allow correlations only restricted by the no-signalling principle---has classical correlations if self-dualized, even though the generalized maximal tensor product includes maximally entangled states. 
We showed how the Spekkens theory is related to this model, since it is also self-dual and violates the no-restriction hypothesis: but were it to satisfy this principle, by taking the full dual cone, it would produce nonlocal correlations.

As another application for restricted models, we show that restrictions can be used to alter theories, such that their measurements are inherently noisy.
This is different to the unrestricted theories, since in our noisy theories it holds that for pure states there is no non-trivial extremal effect occurring with certainty. 
We derive the maximal CHSH violation \cite{CHSH} of a noisy version of boxworld as a function of a noise parameter.

The modified framework is therefore suitable for examining new situations that could not be addressed using the traditional framework.
In particular the self-dualization procedure might be useful for the study of strong self-duality that has recently received much interest \cite{Barnum08,polypaper,Mueller12}.

\begin{acknowledgements}
We thank Jonathan Barrett, Christian Gogolin and Haye Hinrichsen for insightful discussions. PJ is supported by the German Research Foundation (DFG). RL is supported by the Templeton Foundation.
\end{acknowledgements}

\bibliography{refs}   % name your BibTeX data base
\bibliographystyle{apsrev}       % APS-like style for physics

\end{document}